\title{Almost-catalytic Computation} 

\author{Sagar Bisoyi\footnote{Work was done while the the author was a masters student at IIT Madras. Email:~{\tt sagarbisoyi@gmail.com}} \and Krishnamoothy Dinesh\footnote{Indian Institute of Technolgy, Palakkad, India. Email:{~\tt kdinesh@iitpkd.ac.in}} \and Bhabya Deep Rai\footnote{Indian Institute of Technology Madras, Chennai, India. Email: {~\tt \{cs21d200|jayalal\}@cse.iitm.ac.in}, The fourth author's work is also supported by SERB-CRG Grant No : CRG/2020/003553 by Govt of India.} \and Jayalal Sarma$^\ddagger$. }

\documentclass[11pt]{article}

\usepackage{palatino}
\usepackage{setspace}
\usepackage[margin=1in]{geometry}
\usepackage{complexity}
\usepackage{enumerate,cite,url}
\usepackage[justification=raggedright,singlelinecheck=false]{caption}
\usepackage{graphicx}
\usepackage{amsmath,epstopdf}
\usepackage[skip=2pt,font=small,labelfont=bf,textfont=bf]{caption}
\usepackage{subcaption,float}
\usepackage{color}
\usepackage{multirow}
\usepackage{tikz}
\usepackage{etoolbox,array, blindtext}
\usepackage[colorlinks=true,linkcolor=blue,citecolor=blue]{hyperref}
\usepackage[capitalize]{cleveref}
\usepackage[verbose]{backref}

\usepackage{booktabs}
\usepackage{wrapfig}
\usepackage{eso-pic}
\usepackage{forest}
\usepackage{tabto}

\usepackage{amsmath,amssymb,amsfonts,amsthm}

\usepackage{thmtools}
\usepackage{thm-restate}

\declaretheorem[name=Theorem,numberwithin=section]{theorem}

\usepackage[linkcolor=blue,citecolor=blue,colorlinks=true]{hyperref}
\usepackage{color}
\usepackage{caption}
\usepackage{verbatim,graphicx}
\usepackage{complexity}
\usepackage{algorithm}
\usepackage{algorithmic}
\newcommand{\expt}{\mathbb{E}}

\usepackage{tikz}
\usepackage{tikz-qtree}

\usepackage{todonotes}
\newcounter{todocounter}

\theoremstyle{plain}
\newtheorem{lemma}[theorem]{Lemma}
\newtheorem{proposition}[theorem]{Proposition}

\theoremstyle{definition}
\newtheorem{definition}[theorem]{Definition}

\newcommand{\F}{{\mathbb{F}}}

\newcommand\numberthis{\addtocounter{equation}{1}\tag{\theequation}}

\newcommand{\PARITY}{\textrm{\sc Parity}}

\newcommand{\ACL}{\textrm{\sf ACL}}

\newcommand{\CL}{\textrm{\sf CL}}
\renewcommand{\L}{\textrm{\sf L}}
\newcommand{\CNL}{\textrm{\sf CNL}}
\newcommand{\coCNL}{\textrm{\sf coCNL}}
\newcommand{\CSC}{\textrm{\sf CSC}}
\renewcommand{\CSL}{\textrm{\sf CSL}}
\newcommand{\CBPL}{\textrm{\sf CBPL}}
\newcommand{\CUL}{\textrm{\sf CUL}}

\newcommand{\ACSPACE}{\textrm{\sf ACSPACE}}
\newcommand{\CSPACE}{\textrm{\sf CSPACE}}

\newcommand{\calR}{\mathcal{R}}

\newcommand{\calP}{\mathcal{P}}

\usepackage{collect}
\def\movetoappendix{1}

\definecollection{appendix}
\makeatletter
\newenvironment{aproof}[2]
  { \@nameuse{collect}{appendix}
  { \subsection{#1} \label{#2} \begin{proof} } {\end{proof}}
  }{\@nameuse{endcollect}}
\makeatother

\makeatletter
\newenvironment{appsection}[2]
  { \@nameuse{collect}{appendix}
  { \subsection{#1} \label{#2} }
  {}
  }{\@nameuse{endcollect}}
\makeatother

\ifthenelse{\equal{\movetoappendix}{0}}{
        \renewenvironment{aproof}[2]{\begin{proof}} {\end{proof} }
        \renewenvironment{appsection}[2]{} {}
}{}

\begin{document}
\maketitle
\vspace{-6mm}
\begin{abstract}
Designing algorithms for space bounded models with restoration requirements on (most of) the space used by the algorithm is an important challenge posed about the catalytic computation model introduced by Buhrman \textit{et al.} (2014). Motivated by the scenarios where we do not need to restore unless $w$ is {\em useful}, we relax the restoration requirement: only when the content of the catalytic tape is $w \in A \subseteq \Sigma^*$, the catalytic Turing machine needs to restore $w$ at the end of the computation. We define, $\ACL(A)$ to be the class of languages that can be accepted by almost-catalytic Turing machines with respect to $A$ (which we call the catalytic set), that uses at most $c\log n$ work space and $n^c$ catalytic space. We prove the following for the almost-catalytic model.
\begin{itemize}
\item We show that if there are almost-catalytic algorithms for a problem with catalytic set as $A \subseteq \Sigma^*$ and its complement respectively, then the problem can be solved by a zero-error randomized algorithm that runs in expected polynomial time. More formally, for any language $A \subset  eq \Sigma^*$, $\ACL(A) \cap \ACL(\overline{A}) \subseteq \ZPP$. In particular, when $A \in \L$, $\ACL(A) \cap \ACL(\overline{A}) = \CL$. This leads to newer algorithmic approaches for designing catalytic algorithms.
\item Using the above, we derive that to design catalytic algorithms for a language, it suffices to design almost-catalytic algorithms where the catalytic set is the set of strings of odd weight ($\PARITY$). 
Towards this, we consider two complexity measures of the set $A$ which are maximized for $\PARITY$. One is the random projection complexity (denoted by ${\cal R}(A)$) and the other is the subcube partition complexity (denoted by ${\cal P}(A)$). We show that, for all $k \ge 1$, there exists a language $A_k \subseteq \Sigma^* $ such that $\DSPACE(n^k) \subseteq \ACL(A_k)$ where for every $m \ge  1$, $\calR(A_k \cap \{0,1\}^m) \ge \frac{m}{4}$ and $\calP(A_k \cap \{0,1\}^m)=2^{m/4}$. This is in contrast to the catalytic machine model where it is unclear if it can accept all languages in $\DSPACE(\log^{1+\epsilon} n)$ for any $\epsilon > 0$.
\item 
Improving the partition complexity of the catalytic set $A$ further, we show that for all $k \ge 1$, there exists $A_k \subseteq \{0,1\}^*$ such that  $\DSPACE(\log^k n) \subseteq \ACL(A_k)$ where for every $m \ge 1$,  $\calR(A_k \cap \{0,1\}^m) \ge \frac{m}{4}$ and $\calP(A_k \cap \{0,1\}^m)=2^{m/4+\Omega(\log m)}$. Our main new technique for the last two items is the use of error correcting codes to design almost-catalytic algorithms. 
\item We also show that, even when there are more than two alphabet symbols, if the catalytic set $A$ does not use one of the alphabet symbols, then efficient almost-catalytic algorithms with $A$ as the catalytic set can be designed for any language in $\PSPACE$.
\end{itemize}

\end{abstract}
\newpage

\section{Introduction}

The catalytic Turing machine model (originally proposed by \cite{BCKLS14}) involves a Turing machine that is equipped with an input tape, a work tape and a special tape called the catalytic tape. Let $s,c: \mathbb{N} \longrightarrow \mathbb{N}$ be non-decreasing functions. A language $L$ is said to be decided by a \textit{catalytic} Turing machine $M$ in space $s(n)$ and using catalytic space $c(n)$ if on every input $x$ of length $n$ and arbitrary string $w \in \{0,1\}^{c(n)}$ of length $c(n)$ written on the catalytic tape, the machine halts with $w$ on its catalytic tape. During the computation, $M$ uses at most $s(n)$ tape cells on the work tape and $c(n)$ cells on its catalytic tape, and $M$ correctly outputs whether $x \in L$. $\CL$ is the class of languages that can be accepted by catalytic Turing machines that use at most $O(\log n)$ work space, and $O(n^c)$ catalytic space.

In addition to its theoretical appeal, the motivation for this model (c.f.~\cite{BCKLS14},\cite{BKLS18}) also comes from practically relevant contexts - where the memory that algorithms need is all used up to store otherwise useful data. In such situations, catalytic algorithms (and more formally, catalytic Turing machines) that guarantee restoration of the content of their catalytic tape to the original content, are arguably useful.
	
A natural question is whether this extra space (which needs to be restored to its original content at the end of the computation) helps at all. Quite surprisingly, \cite{BCKLS14} showed that $\L$-uniform $\TC^1 \subseteq \CL$. The fact that $\NL \subseteq \TC^1$ makes this immediately surprising for a space complexity theorist, because it implies that the directed graph reachability problem has a deterministic algorithm in the above model that uses $O(\log n)$ space in the worktape and at most $\poly(n)$ space in its catalytic tape.


\cite{BCKLS14} also showed that $\CL$ is contained in $\ZPP$. The main observation that leads to this upper bound is that two computations starting with different initial catalytic tape contents, say $w$ and $w'$ cannot reach the same configuration at any point in their computations on the same input. 
In a subsequent work, \cite{BKLS18} explores the power of non-determinism in catalytic space. $\CNL$ is the class of problems solvable by non-deterministic logspace catalytic Turing machines. Using similar ideas from \cite{BCKLS14}, it was shown in \cite{BKLS18},\cite{DGJST20} that the $\ZPP$ upper-bound holds even for non-deterministic and randomized variants of catalytic logspace classes. 
\cite{BKLS18} showed that under a plausible hardness assumption, $\CNL = \coCNL$. In a work by~\cite{DGJST20}, it is shown that under the same hardness assumption and using very similar techniques, $\CBPL = \CSL = \CSC^1 = \CL$. Recently,~\cite{CLMP24} completely removed the need for any hardness assumption and showed that $\CBPL = \CL$ unconditionally. Here $\CBPL$ and $\CSL$ are sets of languages solvable by logspace randomized and symmetric catalytic Turing machines, respectively. $\CSC^1$ denotes  the set of languages solved by catalytic log-space machines that run in polynomial time. 
\cite{GJST19} showed that under the same hardness assumption $\CNL = \CUL$. For more details, the reader is referred to the following surveys:~\cite{KouckySurvey},\cite{MertzSurvey}. Algorithmic techniques that were used to design catalytic algorithms have also been proven helpful in designing non-trivial space efficient algorithms for the Tree evaluation problem (proposed in~\cite{CMPDBS12}). For more details, see~\cite{CM24} and the references therein.
\\[-4mm]

\noindent {\bf Our Results:}
Motivated by the scenarios where we do not need to restore unless $w$ is {\em useful}, we relax the restoration requirement: only when the content of the catalytic tape is $w \in A \subseteq \Sigma^*$, the catalytic Turing machine needs to restore $w$ at the end of the computation.  Indeed, $A \subseteq \Sigma^*$ represents the set of ``useful'' $w$'s. We call such Turing machines as \textit{almost-catalytic Turing machines} and the languages accepted by such machines, using logarithmic work space and polynomial catalytic space as $\ACL(A)$~(See ~\cref{sec:prelims} for a formal definition). We call the set $A$ to be the \textit{catalytic set}.

Thus, the major challenge in this context is to design algorithms for useful catalytic sets. We first consider two ways of exploring the almost-catalytic in terms of the catalytic set $A$. Firstly in terms of the cardinality of $A$ and secondly in terms of the complexity of $A$. To start with, observe that $\forall A \subseteq \Sigma^*$, $\ACL(A) \subseteq \PSPACE$. In addition, it is easy to observe that $\ACL(\Sigma^*) = \CL$ and $\ACL(\emptyset) = \PSPACE$. Given this, one natural way to work towards catalytic logspace algorithms for $\PSPACE$ from almost-catalytic algorithms is to parameterize based on the size of $A$. Defining $f(n) = |A \cap \{0,1\}^n|$ to be a measure of sparsity of $A$, we are interested to see how close can the function $f(n)$ be to $2^n$, such that we have almost-catalytic algorithms for every language in $\PSPACE$. 

In this direction, it is easy to see that if $A$ is a tally set ($A \subseteq \{1\}^*$), then $\PSPACE = \ACL(A)$. Such a consequence is unclear if $A$ is only known to be polynomially sparse. However, if $A$ is polynomially sparse with low space complexity, then, we can simulate the whole of $\PSPACE$ using almost-catalytic Turing machines.
That is, for any sparse set $A \in \L$, $\ACL(A) = \PSPACE$ (see \cref{prop:ACA-sparse-PSPACE}). Indeed, it is more challenging to design $\ACL(A)$ algorithms for every language in $\PSPACE$ when $A$ is large in size.

However, we note that there is a set $A$ with exponential density for which we can design almost-catalytic algorithms to accept any language in $\DSPACE(n^k)$ (see~\cref{thm:code-improved-catalytic-reviewer}). This implies that $|A\cap \{0,1\}^n|$ is not a good parameter to measure our progress towards designing catalytic algorithms by this approach.

To make further progress, we turn to the structural front. We show a limitation of the almost-catalytic Turing machines with respect to $A$ by showing the following upper bound.

\begin{restatable}[]{theorem}{ACLZPPandL}\label{thm:ACL-ZPP-L}
For any $A \subseteq \Sigma^*$, it holds that $\ACL(A) \cap \ACL(\overline{A}) \subseteq \ZPP$. If $A \in \L$ then $\ACL(A) \cap \ACL(\overline{A}) = \CL$.
\end{restatable}

The first part of the above theorem and the argument is a generalization of the idea in \cite{BCKLS14} which shows $\CL \subseteq \ZPP$.	In particular, when $A=\Sigma^*$ or $A=\emptyset$, we recover their result. We remark that this generalization is different from the compress-or-random method that appears in~\cite{CLMP24,P23}. We also remark that, unlike the arguments in \cite{BCKLS14}, the Theorem~\ref{thm:ACL-ZPP-L} or the proof of it, does not imply for any almost-catalytic Turing machine runs in expected polynomial time.
The second part of the above theorem can also be viewed as a method of obtaining catalytic algorithms by designing almost-catalytic algorithms with respect to an appropriate set $A$. 

A notable example of such a set is the language $\PARITY$ consisting of strings over $\{0,1\}^*$ with an odd number of ones. Indeed, $\PARITY \in \L$. However, if we have an almost-catalytic algorithm for a language $L$ with the catalytic set being  $\PARITY$, then there is an almost-catalytic algorithm with respect to $\overline{\PARITY}$ as well (See Proposition~\ref{prop:parity-almost-paritybar}). Hence, $\ACL(\PARITY) = \CL$. Thus, it suffices to design almost-catalytic algorithms with respect to $\PARITY$ and we set this as the target.

To measure our progress towards the set $\PARITY$, we define two measures for the set $A$, defined below, which are maximized for parity.

\begin{description}
\item{\bf Random Projection Complexity:} For an $A \subseteq \{0,1\}^m$ , we define, for an $\epsilon \ge 0$, the random projection complexity, ${\calR}_{\epsilon}(A)$ as the largest $\ell \geq 0 $ such that:  
$\Pr_{\substack{ T \subseteq [m] \\ |T| = \ell}} \left[ \left| A_T \right| \geq 2^{\ell-1} \right] \geq 1-\epsilon$
where $A_T$ denotes the set of strings in $A$ projected to the indices in $T$.
Observe that ${\calR}_{0}(\PARITY_m) = m-1$. Thus, in order to approach $A=\PARITY$, we will design almost-catalytic computation with respect to set $A$, where $\calR_\epsilon(A)$ is as large as possible where $\epsilon$ is close to $0$, say $2^{-\alpha m}$ for some small constant $0 \le \alpha< 1$. In this case, we use $\calR(A)$ to denote $\calR_{2^{-\alpha m}}(A)$.\\

\item{\bf Subcube Partition Complexity:} 
A subcube $C$ of the cube $\{0,1\}^m$ is given by a mapping (partial assignment) $\alpha : [n] \to \{0,1,*\}$ and is defined to be the set of all vectors in the Boolean hypercube on $n$ bits, $B_n$, that agree with $\alpha$ on coordinates that are assigned a non-$*$ value by $\alpha$. More precisely the subcube $C_\alpha$ is the set $\{x \in \{0,1\}^m : \alpha(i) \ne * \implies
x_i = \alpha(i)\}$. 
For a set $A$, a partition $C = \{C_1, \ldots, C_t\}$ of $A$ into subcubes $C_i$ such that $C_i \subseteq A$ is called a subcube partition of $A$. We denote by $\calP(A)$ the minimum number of subcubes in a subcube partition of $A$.
Observe that $\calP(\PARITY_m) = 2^{m-1}$. Thus, in order to approach $A = \PARITY$, we propose to design almost-catalytic algorithms for sets with high partition complexity.
\end{description}

We remark about the choice of the above two measures in our journey towards achieving $\PARITY$ as our catalytic set. As noted earlier, there are specific catalytic sets (see \cref{thm:code-improved-catalytic-reviewer}) $A \subseteq \{0,1\}^*$ which are of exponential density for which $\PSPACE \subseteq \ACL(A)$. However, it can be shown (see~\cref{prop:measures-for-reviewers-A}) that this catalytic set $A$ has a subcube partition complexity $\calP(A)$ of $1$ and small random projection complexity $\calR(A)$. Hence, it is natural to look for other catalytic sets $A$ such that $\ACL(A)$ is powerful enough to simulate polynomial space bounded computation, and which possess larger values for one or both of these measures.

As our next result, we show the following simulation of $\DSPACE(n^k)$  almost-catalytically for set $A$ with a large $\calP(A)$ and $\calR(A)$.

\begin{restatable}[]{theorem}{almostCatalytic} \label{Spthm}
For all $k \ge 1$, there exists a language $A_k \subseteq \{0,1\}^* $ such that $\DSPACE(n^k) \subseteq \ACL(A_k)$ where for every $m \ge  1$, $\calR(A_k \cap \{0,1\}^m) \ge \frac{m}{4}$ and $\calP(A_k \cap \{0,1\}^m)=2^{m/4}$.
\end{restatable}

An important challenge in designing catalytic algorithms is the incompressibility of the string $w$. However, in our context, the set $A_k$ in the above theorem may be viewed as compressible since the set of codewords can be represented by the set of messages. But note that this compressibility is not directly useful for designing the almost-catalytic algorithm since the message length can also be linear in $n$ and hence cannot be stored in the logarithmic work space.

Going further, when we need to simulate only polylogarithmic space, the partition complexity of the catalytic set $A$ for which we restore can be improved. We prove the following theorem in this direction:


\begin{restatable}[]{theorem}{codeCatalytic}
\label{thm:code-improved-catalytic}
For all $k \ge 1$, there exists $A_k \subseteq \{0,1\}^*$ such that  $\DSPACE(\log^k n) \subseteq \ACL(A_k)$ where for every $m \ge 1$,  $\calR(A_k \cap \{0,1\}^m) \ge \frac{m}{4}$ and $\calP(A_k \cap \{0,1\}^m)=2^{m/4+\Omega(\log m)}$.
\end{restatable}

We remark that this is in contrast to the catalytic machine model where it is unclear if it can accept all languages in $\DSPACE(\log^{1+\epsilon} n)$ for any $\epsilon > 0$.
At the other end of the spectrum, note that, if \cref{thm:code-improved-catalytic} holds when $A_k \cap \{0,1\}^m$  covers the whole of $\{0,1\}^m$ (or even the set $\PARITY$), then it would imply that $\DSPACE(\log^k n) \subseteq \CL$. Since $\CL  \subseteq \ZPP$\cite{BCKLS14}, this would show that $\DSPACE(\log^k n) \subseteq \ZPP$, which in-turn would separate $\L$ and $\ZPP$ by the space hierarchy theorem. 

Exploring the power of additional alphabets, we show that even if the catalytic tape alphabet has even a single symbol that is not included in the alphabet for the catalytic set (irrespective of the size), the almost-catalytic machine can simulate the whole of $\PSPACE$ (See Proposition~\ref{prop:3-size-alpha}).


\paragraph*{\bf Our Techniques:} 
Our technique starts with a novel approach towards designing almost-catalytic algorithms using codes that can be decoded space efficiently. At a high level, the idea is as follows: let us say we want to design a catalytic Turing machine accepting a language $L$ which has a Turing machine that runs in space $c(n)$. For the catalytic Turing machine, the given content of the catalytic tape can be treated as the codeword (for a fixed code), the Turing machine proceeds to modify the content of the tape according to its computational needs.  The modification of the work tape during computation can be seen as introducing ``errors'' to the codeword. Finally we use the decoding algorithm to correct the ``errors'' and finally obtain the original codeword we started with, thus achieving the restoration condition.
	
Indeed, there are a number of challenges in implementing the above plan. The first limitation is that the number of bits modified to the initial string on the catalytic tape must be such that the modified string (after computation) is still within a decodable distance from the original word. Thus, if $c(n)$ is the catalytic space available, we can hope to allow the catalytic TM to use only strictly $o(c(n))$ bits in the catalytic tape during the computation. Thus, an interesting target is to simulate normal Turing machines that use an asymptotically smaller amount of space. 

A second challenge is that the code must be decodable in deterministic logarithmic space, as we have only so much work space. Fortunately, there are codes that have constant rate and constant relative distance, for which logspace decoding algorithms are known (See Theorem 19 in \cite{linearSpielman} and Theorem 14 in \cite{GK06}). Using additional decodability properties of Spielman codes, we show that the set $A$ can be expanded to achieve larger random projection complexity and larger subcube partition complexity, thus progressing towards $A = \PARITY$. In order to establish the progress in terms of measures of the catalytic sets, we also employ techniques from basic combinatorics of codes, and Fourier analysis of Boolean functions to estimate the partition complexity of the catalytic set in our algorithms.
	
	

 \paragraph{Related Work:} In a simultaneous work, \cite{GJST24}, considers the model of \textit{lossy catalytic computation} which is a catalytic Turing machine with an $O(\log n)$-sized work tape and polynomial-sized catalytic tape where the restoration condition is weaker - only \textit{all except a constant number of bits} are needed to be restored. They show that this relaxation does not add any power to the model - such catalytic Turing machines can only accept languages accepted by standard catalytic Turing machines with the same amount of catalytic space and work space. The technique that they use is to hash the first bits of the configuration spaces when the number of changes are limited. We remark that this is incomparable with the relaxation that we impose where for some strings (strings outside the set $A$) it is not even needed to be restored, while for some other strings (that is, strings in $A$), they need to be restored without a loss. Our techniques and motivating questions are also different from the work of \cite{GJST24}.

	\section{Preliminaries} \label{sec:prelims}

    We begin by defining catalytic computation as described by \cite{BCKLS14}. We refer the reader to standard references~\cite{Goldreich_2008,arora-barak} for definitions of the complexity classes not defined in this paper.

    A catalytic Turing machine is a Turing machine with a read-only input tape, a work tape of size $s(n)$, and a catalytic tape of size $c(n)$ initially containing some $w \in \{0,1\}^{c(n)}$, where $n$ is the size of the input. The machine  $M$ is said to decide a language $L$ if (1)~$x\in L$ if and only if $M$ accepts on input $x$ for all possible initial catalytic content $w$ and (2)~For each input $x \in \{0,1\}^n$ and any initial catalytic tape content $w$, $M$ halts with $w$ on its catalytic tape.
    We shall use the term \textit{catalytic space} to denote the space in the catalytic tape. The class $\CSPACE(s(n), c(n))$ is the set of all languages decided by a catalytic Turing machine with work space $s(n)$ and catalytic space $c(n)$. The class $\CL$ denotes $\CSPACE(O(\log n), \poly(n))$.

    

\subsection{Bounds on the Complexity Measures}
    Recall, from the introduction, that for a set $A \subseteq \{0,1\}^m$, we use $\mathcal{R}(A)$ to denote its Random projection complexity and $\mathcal{P}(A)$ to denote its Subcube partition complexity. For an integer $b$ dividing $m$, let $A_b= \{w \mid w \text{ is of the form } 0^{m/b}(0+1)^{m-m/b}\} \subseteq \{0,1\}^m$. 
    
    \begin{proposition}       
     \label{prop:measures-for-reviewers-A}
        For the set $A_b$ defined above, $\mathcal{P}(A_b)=1$ and for a constant $b$, $\mathcal{R}(A_b) = O(1)$.
    \end{proposition}
    \begin{proof}
    The subcube partition complexity $\mathcal{P}(A_b)$ is $1$ as the entire set is contained in the subcube $C$ with the first $m/b$ coordinates set to $0$ and it cannot be any smaller. 
    
    We need an upper bound on the random projection complexity $\mathcal{R}(A_b)$. Towards this, let $\ell$ be the smallest value for which $\Pr_{\substack{ T \subseteq [m] \\ |T| = \ell}} \left[ \left| (A_b)_T \right| \geq 2^{\ell-1} \right] < 1-\epsilon$, where $\epsilon$ is a small constant. It can be seen that this probability is lower bounded by $\alpha^{\ell}$ for a constant $\alpha$ that depends on $b$. Hence for a constant $b$, $\ell$ is bounded by a constant.
    \qed \end{proof}

We now establish a lower bound for the measure for another set $A$ which we use as a catalytic set later. We quickly recall linear codes, and related parameters below.


    A linear code over a $q$-ary alphabet of length $m$ and dimension $k$ is a linear subspace $C$ with dimension $k$ of the vector space $\mathbb{F}_{q}^{m}$. The distance $d$ of a linear code $C$ is the minimum Hamming distance between any two codewords in $C$, where Hamming distance between two codewords is the number of locations where they differ. Furthermore, $C$ is said to be an $[m,k,d]_q$ code if it has length $m$, dimension $k$, distance $d$ and alphabet size $q$. The relative distance of a  $[m,k,d]_q$ code, $\delta$ is defined as $\delta = \frac{d}{m}$. The covering radius of a code $C$ is the minimum $D$ such that for all $w \in \F_q^m$ there exists a codeword $c \in C$ such that $d(c,w) \le D$. We will have the following bounds on the measure.

        
    \begin{proposition}       \label{lem:codes}
        If $A$ is a set of codewords for an $[m, k, \delta m]_2$ code with $\delta$ being a constant, then $\mathcal{R}_\epsilon (A) \ge k$ for $\epsilon=2^{-2k}$.
    \end{proposition}
       The proof of the above Lemma is a standard application of codes. 
       We reproduce the argument in Appendix~\ref{app:lem:codes} for completeness.
       
       \begin{aproof}{Proof of Lemma~\ref{lem:codes}}{app:lem:codes}
        





        Let $E: \{0,1\}^k \rightarrow \{0,1\}^m$ be the encoding associated with the given $C=(m, k, \delta m)_2$ code. Let $\ell$ be the random projection complexity of the set of codewords. We show that when $\epsilon=1/2^{2k}$, $\calR_{\epsilon}(C) \ge k$.
        
        Consider any set $S \subseteq \{0,1\}^k$ such that $|S| \ge 2^{\ell-1}$.  We will fix $S$ later. Firstly notice that, 
        \[\Pr_{\substack{T \subseteq [m]\\ |T|=\ell}} \big[ \forall x \neq y \in S \textrm{ such that } E(x)_T \neq E(y)_T \big ] \leq  \Pr_{\substack{T \subseteq [m]\\ |T|=\ell}} \big [ \big |A|_T \cap \{0,1\}^\ell \big | \geq 2^{\ell-1} \big ]\]
        Our goal is to show a lower bound of $1-\epsilon$ for the term in the left-hand size. Instead, we start by analyzing the complementary event and show that its probability is upper bounded by $\epsilon$. For any distinct pair $x,y \in S$, 
        \[\Pr_{\substack{T \subseteq [m]\\ |T|=\ell}} \big[ E(x)_T = E(y)_T \big ] \leq (1 - \delta)^\ell\]

        The above follows from the fact that since $A$ is a code of distance $\delta m$, the probability for $E(x)$ and $E(y)$ to be the same at a random index in $T$ is at most $1-\delta$. Thus, we have, 
        \[\Pr_{\substack{T \subseteq [m]\\ |T|=\ell}} \big[ \exists x \neq y \in S \textrm{ such that } E(x)_T = E(y)_T \big ] \leq (1 - \delta)^\ell {\binom{|S|}{ 2}}\]

        Choosing $S$ to be any subset of $\{0,1\}^k$ of size $2^{k-1}$, we want $(1-\delta)^\ell \binom{2^{k-1}}{2} \le \epsilon$. This means that 
        \[ \ell \ge \frac{2k-\log (1/\epsilon)}{\log (1/(1-\delta))}\]

        In addition, since $|S| \ge 2^{\ell-1}$, we have $k \ge \ell$. For our choice of $\epsilon$ all values of $\ell$ up to $k$ are feasible. Since we need the maximum possible $\ell$, we choose $\ell=k$. This completes the proof.
    \end{aproof}
    
	\subsection{A Lower Bound on the partition complexity for Union of Hamming Balls} \label{ssec:lb-hamming-balls}
    As a part of showing improved partition complexity lower bound in \cref{thm:code-improved-catalytic}, in this section, we outline the tools and ideas from the area of Fourier representation of Boolean functions that we used. The reader is referred to \cite{AOBF} for a comprehensive background on this subject.

    For two strings, $x,y \in \{0,1\}^m$, the Hamming distance, denoted by $\Delta(x,y)$, is the number of locations in which $x$ and $y$ differ. The same definition can be extended to subsets as follows: for any $A,B \subseteq \{0,1\}^m$, $\Delta(A,B) = \min\{\Delta(a,b) \mid a \in A, b \in B \}$. A set $H \subseteq \{0,1\}^m$ is said to be a Hamming ball if and only if there exists a $k \ge 0$ and a $z \in \{0,1\}^n$ such that for every $h \in H$, $\Delta(h,z) \le k$. We call $k$ as the \textit{radius} of the Hamming ball $H$ and $z$ to be its \textit{center}.
    
    The catalytic set considered in \cref{thm:code-improved-catalytic} is a union of Hamming balls centered on logspace decodable codewords. As a first step, we show that the partition complexity of this set is precisely the sum of partition complexity of the individual Hamming balls.
    
    \begin{proposition}(See \cref{prop:partition-complexity-balls-union}, \cref{app:sec:lbunion})
        Let $A \subseteq \{0,1\}^*$ be such that for any $m \ge 1$, $A_m := A \cap \{0,1\}^m$ can be expressed as a union of Hamming balls $H_1, H_2, \ldots, H_t$ over $\{0,1\}^m$ such that for any $i\ne j$, $\Delta(H_i, H_j) > 1$. Then, $\calP(A_m) = \sum_{i=1}^t \calP(H_i)$.
    \end{proposition} 
    
    Define the Boolean function $Th_{m,k} \colon \{0,1\}^m   \to \{-1,1\}$ as for any $x \in \{0,1\}^m, Th_{m,k}(x) = -1$ if $|x| \le k$ and $1$ if $|x| > k$.  
 
    Now, it remains to compute the partition complexity of a Hamming ball. The starting observation is that a Hamming ball of radius $k$ centered at $0^m$  is precisely the set of inputs on which the threshold Boolean function $Th_{m,k}$ evaluates to $-1$. The next observation, due to \cite{CKLS16} (Lemma 3.8), is that the partition complexity of a set viewed as a Boolean function is lower bounded by the sum of absolute values of its Fourier coefficients.

    In \cref{prop:fc-threshold} and \cref{prop:fc-values-more-than-k} (both appearing in \cref{app:sec:lbunion}), we obtain closed-form expressions for Fourier coefficients of $Th_{n,k}$. Using this, we show the following lower bound on the partition complexity of a Hamming ball.
    
    \begin{proposition}(See \cref{prop:partition-complexity-ball}, \cref{app:sec:lbunion})
        Let $H$ be a Hamming ball over $\{0,1\}^m$ of radius $k < m/2-\sqrt{m}$ centered at $0^m$. Then $\calP(H) = \Omega(k)$.
    \end{proposition}

    The main lemma that is used in~\cref{sec:limalmostcatalytic} for arguing the improved bound on subcube complexity in \cref{thm:code-improved-catalytic} is the following. 

    \begin{lemma}(See \cref{lem:partition-complexity-ball-main}, \cref{app:sec:lbunion})
        Let $A \subseteq \{0,1\}^*$ such that for every $m \ge 1$, $A_m$ is a disjoint union of Hamming balls $H_1, \ldots, H_t$ of radius $k<m/2-\sqrt{m}$ over $\{0,1\}^m$ such that for every $i,j \in [t]$, $\Delta(H_i, H_j) > 1$. Then for every $m\ge 1$, $\calP(A_m) = \Omega(tk)$.
    \end{lemma}

    Due to space constraints, detailed proofs of these statements are moved to \cref{app:sec:lbunion}.
    
\section{Almost-catalytic Turing Machines}\label{sec:almostcatalytic}

In this section, we present the definition and our results on Almost-catalytic Turing machines. We begin with the following definition.

   \begin{definition}[{\bf Almost-catalytic Computation with respect to $A$ : $\ACSPACE_A$ and $\ACL(A)$}] \label{def:almost-cat-space}
    Let $A \subseteq \Sigma^*$, a language $L$ is said to be in the class $\ACSPACE_A(s(n), c(n))$ if there is a Turing machine $M$ which on inputs of length $n$ uses a work tape of size $s(n)$ and catalytic tape of size $c(n)$ (over an alphabet set of size $2$) such that, (1)~for all $x \in \Sigma^*$, $x \in L$ if and only if the Turing machine $M$ accepts $x$. (2)~for all $w \in A$, if the machine $M$ starts the computation with content of the catalytic tape as $w$, then at the end of the computation $w$ will be restored back in the tape. For all $w \not \in A$, the algorithm need not restore the catalytic tape.
    
    Furthermore we define $\ACL(A)$ to denote the class $\ACSPACE_A(O(\log n), O(n^c))$ for some constant $c$.
    \end{definition}

We  make some preliminary observations about almost-catalytic computation.
Indeed, by definition, $\CL = \ACL(\Sigma^*)$, and $\PSPACE = \ACL(\emptyset)$. In general, for any $A \subseteq \Sigma^*$, $\CL \subseteq \ACL(A) \subseteq \PSPACE$. Moreover, there are languages $A \subsetneq \Sigma^*$  for which the $\ACL(A)$ can simulate the whole of $\PSPACE$.
The following proposition is also easy to see.
\begin{proposition} 
If $A = \{ 1^n \mid n \geq 0 \}$, then $\PSPACE = \ACL(A)$    
\end{proposition}
The above proposition is true since the catalytic tape can be filled with $1^n$ at the end of the computation irrespective of the original content. However, it is a challenge to show the above for an arbitrary singleton set $A$. 

    A natural question is about the density of the catalytic set. We establish that for every $k$, there are sets with high density with respect to which every language in $\DSPACE(n^k)$ admits almost-catalytic algorithms.
    \begin{proposition}\label{thm:code-improved-catalytic-reviewer}
        For any $k \ge 1$, there exists a language $A \subseteq \{0,1\}^*$ with $\DSPACE(n^k) \subseteq \ACL(A)$ via an almost-catalytic logspace machine using $m=bn^k$ catalytic space for some constant $b \ge 1$, such that for any $m \ge 1$, $|A \cap \{0,1\}^m| \ge 2^{m-m/b}$.
    \end{proposition}
    \begin{proof}
        Consider a language $L$ that can be decided by a machine $M$ in $n^k$ space. Now, we shall construct an almost-catalytic Turing machine $M'$ deciding $L$ using $m = b n^k$ catalytic space for some $b \geq 1$. We shall define the catalytic set $A$ used by $M$' as $\{w \mid w \text{ is of the form } 0^{n^k}(0+1)^{(b-1)n^k} \mid n \ge 1\}$.

        The machine $M'$ works as follows: Simulate $M$ on input $x$ using the first $n^k$ many bits of the catalytic tape. Now, for restoration, we set the first $m$ many bits back to $0$, which belongs to the set $A$. Finally, we observe that $|A \cap \{0,1\}^m| = 2^{(b-1)n^k} = 2^{m-m/b}$.
        \qed
    \end{proof}
 
	At the other extreme, if $|A\cap \{0,1\}^n| = \poly(n)$ i.e. $A$ is sparse, we ask the question if it is true that for all sparse $A$, $\ACL(A) = \PSPACE$? We observe that the answer is affirmative when the sparse set $A$ under consideration is in $\L$. 
	
	\begin{proposition}\label{prop:ACA-sparse-PSPACE}
		Let $A \subseteq \Sigma^*$ be a language in $\L$. Then if $A$ is sparse then $\ACL(A) = \PSPACE$.
	\end{proposition}
	\begin{aproof}{Proof of Proposition \ref{prop:ACA-sparse-PSPACE}}{app:prop:ACA-sparse-PSPACE}
		Let $L \in \PSPACE$ via a Turing machine $M_L$. We want to show that we can construct a $\ACL(A)$ machine $M'$ such that it decides $L$. Say $A \in \L$ via the machine $M_A$. Following is the description of $M'$ (Algorithm \ref{log_sparse}) with catalytic tape initialized with $w \in \{0,1\}^{\poly(n)}$:
		
		\begin{algorithm}[htp!]
			\begin{algorithmic}[1]
				\caption{Machine $M'$ on $x \in \{0,1\}^n$ and $w \in \{0,1\}^{\poly(n)}$}\label{log_sparse}
				\STATE Check if $w \in A$ using the work space to run the machine $M_A$. If not, run the machine $M_L$ on the catalytic space. Accept if $M_L$ accepts, Rejects if $M_L$ rejects.
				\STATE Initialize $count=0$
				\REPEAT 
				\STATE Run machine $M_A$ on $w$:
				\IF {$w \in A$}
				\STATE Increment $count$. 
				\STATE Update $w = w - 1$. \hfill 
				\ENDIF
				\UNTIL {$w$ becomes all $0$'s}
				\STATE Run $M_L$ on catalytic tape. If $M_L$ accepts, set $flag = true$, else $flag = false$
				\REPEAT 
                \STATE Run machine $M_A$ on $w$:
				\IF {$w \in A$}
				\STATE Decrement $count$.
				\STATE Update $w=w+1$
				\ENDIF
				\UNTIL $count=0$
				
				\STATE If $flag=true$ then Accept, otherwise Reject and halt.
			\end{algorithmic}
		\end{algorithm}
		
		    Because $A$ is in $\L$, we can compute the membership of $w$ in $A$ using only the work space, which is logarithmic in size. If $w \not\in A$, it is not essential to restore the catalytic tape hence we simply run the machine $M_L$ on the catalytic tape without restoring $w$. Next, if $A$ is sparse there are only $\poly(n)$ many strings that the machine $M_A$ would accept. A counter that remembers the position of such a string in a lexicographically ordered $A$, would need only $O(\log n)$ many bits for its storage. So if $w \in A$, we start ``decrementing" the string while incrementing the counter, until when $w$ becomes all $0$'s, $count$ stores exactly the position of the string in a lexicographically ordered $A$. Finally, we can simply run the machine $M_L$ on the catalytic tape, and knowing the position of $w$ (in the lexicographically ordered $A$) stored by $count$ helps us restore $w$ at the end.
	\end{aproof}

    The proof for the above theorem can be found in Appendix \ref{app:prop:ACA-sparse-PSPACE}.
    We now show the following proposition for $\PARITY$ which is logspace decidable but is not sparse.
    \begin{proposition}\label{prop:parity-almost-paritybar}
    $\ACL(\PARITY) = \ACL(\overline{\PARITY})$.    
    \end{proposition}
    Along with~\cref{thm:ACL-ZPP-L}, this shows that it suffices to design almost-catalytic logspace algorithms for $A=\PARITY$ to show membership in $\CL$. 
\begin{proof}[Proof of~\cref{prop:parity-almost-paritybar}]
It suffices to show that if $L \in \ACL(\PARITY)$, then $L \in \ACL(\overline{\PARITY})$. Let $L \in \ACL(\PARITY)$ via an almost-catalytic machine $M$. Consider an almost-catalytic machine $M'$ which works by first checking if the catalytic content $w$ belongs to $\overline{\PARITY}$. If yes, it flips the first bit of the catalytic content (which makes the catalytic content to be in $\PARITY$), runs $M$, flips the first bit of catalytic tape and accepts iff $M$ accepts x. The simulation of $M$ will correctly decide $L$ and restore the catalytic content which is the same as $w$ except for the first bit. The final step of $M'$ will restore the first bit. Hence $M'$ restores all strings in $\overline{\PARITY}$ and accepts $L$.
\qed \end{proof}

It is important that the definition of almost-catalytic space (Definition~\ref{def:almost-cat-space}) uses a catalytic tape alphabet set of size $2$. A larger alphabet set can dramatically increase the power of almost-catalytic space. Suppose we let the almost-catalytic machine with catalytic alphabet over a larger $\Gamma$ with $ \{0,1\} \subsetneq \Gamma$ and make the machine restore any set  $A\subseteq \{0,1\}^*$. More precisely, let $\ACL^\Gamma (A)$ denote the languages decidable by almost-catalytic logspace machines working over the catalytic tape alphabet $\Gamma$ with $A \subseteq \Gamma^*$ as the catalytic set. Observe that for any $A \subseteq \Gamma^*$, $\ACL^{\Gamma}(A) \subseteq \PSPACE$. 

We now show that even if the catalytic tape alphabet has even a single symbol that is not included in the alphabet for the catalytic set, the almost-catalytic machine can simulate the whole of $\PSPACE$.

\begin{proposition}\label{prop:3-size-alpha}
    Let $\Sigma$ be an input alphabet set and $\Gamma$ be a catalytic tape alphabet with $|\Gamma|>|\Sigma|$. Then for any $A \subseteq \Sigma^*$, $\PSPACE = \ACL^\Gamma(A)$. In particular, for $\Sigma = \{0,1\}$ and any $\Gamma$ with $|\Gamma| \ge 3$, $\PSPACE = \ACL^\Gamma(\Sigma^*)$
\end{proposition}

\begin{proof}
    Without loss of generality, assume $\Sigma=\{0,1\}$ by suitably fixing a binary encoding for the input alphabets. Let $A\subseteq \Sigma^*$.  It suffices to show that $\PSPACE \subseteq \ACL^{\Gamma}(A)$. 

   Consider a language $L$ in $\PSPACE$ via a $p(n)$ space bounded deterministic Turing machine $M$ where $p(n)$ is a fixed polynomial in $n$. Also, without loss of generality, let the work tape of $M$ use the alphabet set $\{0,1,\text{\textvisiblespace}\}$.

   An almost-catalytic machine $M'$ using catalytic tape alphabet $\Gamma$ having $\{0,1,\widehat{0}\} \subseteq \Gamma$ accepting $L$ with a catalytic tape of length $4p(n)$ is described as follows: 
   Scan across the catalytic tape and check if the initial catalytic content $w$ contains a $\widehat{0}$ symbol. 
   If there is no occurrence of $\widehat{0}$, then $w$ can be a member of $A$ and in particular consists of $1$s and $0$s alone. Using the work tape, $M'$ counts the number of $0$s in $w$ denoted by $m$. 
   
   We now describe how $M'$ simulates $M$. Suppose that the number of $0$s is more than the number of $1$s. Then $m \ge \frac{1}{2} \times 4p(n) = 2p(n)$. The machine $M'$ uses the first $2p(n)$ cells out of the $m$ cells containing $0$ of the catalytic tape to simulate the workspace of $M$. Note that $M$ is over alphabet set $\{0,1,\text{\textvisiblespace}\}$ while the catalytic tape of $M'$ is over the alphabet set $\Gamma$. To handle the work tape symbols of $M$ correctly during the simulation, $M'$ uses the following encoding $E: \{0,1,\text{\textvisiblespace}\} \to \Gamma$ defined as $E(0) = 00$, $E(1) = 0\widehat{0}$ and $E(\text{\textvisiblespace}) = \widehat{0}0$. More precisely, if $M$ reads (or writes) a symbol $\alpha \in \{0,1,\text{\textvisiblespace}\}$ at position $i$ of its tape, $M'$ proceeds to read (or write) $E(\alpha)$ at the $2i$ and $2i+1$th cells having $0$ or $\widehat{0}$ counted from the left end on the catalytic tape. Once the computation ends, restoration of $w$ is achieved (irrespective of whether $w\in A$ or not) by replacing all the $\widehat{0}$ with $0$ at the end of the simulation. Now, if the number of $1$'s are more than the number of $0$s, then $M'$ uses an encoding $E(0) = 11$, $E(1) = 1\widehat{0}$ and $E(\text{\textvisiblespace})=\widehat{0}1$ and repeat the above simulation of $M$ with $0$ replaced by $1$ in the above text. 

    If $w$ contains a $\widehat{0}$, then $w \not\in A$ and therefore $M'$ is not required to restore the catalytic tape. In such a case, $M'$ erases the catalytic tape and simulates $M$ on it. 
    Clearly, if $M$ uses $\poly(n)$ workspace, $M'$ can simulate $M$ using $O(\log n)$ work space and  $4\cdot \poly(n)$ catalytic space. Hence $L \in \ACL^{\Gamma}(A)$ which completes the proof.
\qed \end{proof}

  
        \section{An Upper Bound on Almost-catalytic Computation}
        In this section, we show that for any language $A \subseteq \Sigma^*$, languages computable by almost-catalytic Turing machines with respect to $A$ which are also computable by catalytic Turing machines with respect to $\overline{A}$ are contained in $\ZPP$.

 	\begin{lemma}\label{cat-configs}
		Define for any almost-catalytic Turing machine M restricted to A, $C_t(x,w)$ to be the configuration with input $x$ and catalytic tape content $w$ at time $t$. For all $x$, for all $w,w' \in A$ such that $w \neq w'$ and for all $t,t'$ $C_t(x,w) \neq C_{t'}(x,w')$.
	\end{lemma}
	
	\begin{proof}
		Assume there exists an $x$ and there exists $w,w' \in A$ such that $C_t(x,w) = C_{t'}(x,w')$. Now, from that point onward the computation would be the same and the restoration part would be incorrect for one of $w$ or $w'$, a contradiction. This justifies our lemma.
	\qed \end{proof}
	
 \ACLZPPandL*
	
	\begin{proof}
        First, we argue that for any $A \subseteq \Sigma^* $, $\ACL(A) \cap \ACL(\overline{A}) \subseteq \ZPP$.        
		Let $L \in \ACL(A)$ via Turing machine $M_1$ and $L \in \ACL(\overline{A})$ via Turing machine $M_2$. Algorithm \ref{zpp_alg} describes the $\ZPP$ machine $M'$ for $L$.
		
		\begin{algorithm}
			\caption{\textbf{Description for Machine $M'$} on input $x$ and initial catalytic tape content $w$}\label{zpp_alg}
			\begin{algorithmic}[1]
				\STATE Choose a $w \in \{0,1\}^{\poly(n)}$ u.a.r.
                \STATE Perform steps (3) and (4) in a time shared fashion till one of them halts
				\STATE Run $M_1$ on $x$ with $w$ on a catalytic tape
				\STATE Run $M_2$ on $x$ with $w$ on a separate catalytic tape
                \STATE Accept if and only if the machine that halted accepted.
			\end{algorithmic}
		\end{algorithm}
		
		Correctness follows since $M'$ either simulates $M_1$ or $M_2$ both of which correctly accepts $L$.
  
		We now analyze the run time of $M'$ and show that it runs in expected polynomial time (w.r.t. $w$). Let $t(x,w)$ denote the total number of steps $M'$ makes on $x$ and $w$. Let $t_1(x,w)$ denote the running time of machine $M_1$ on input $w$ in Step 4. Let $t_2(x,w)$ denote the number of steps taken in Step 6. Observe that $t(x,w) = O(\min\{t_1(x,w), t_2(x,w)\})$. 	
		For a fixed $x$,  the expected running time (over the random choices of $w$) of $M'$ can be obtained as $\expt[t(x,w)]  = \expt[\min \{t_1(x,w), t_2(x,w)]$. We now bound the expectation. 
		\begin{align*}
			\expt[t(x,w)] & = \expt[t(x,w)|w \in A] \times \Pr[w \in A] + \expt[t(x,w)|w \in \overline{A}] \times \Pr[w \in \overline{A}] \\
            & \le  \expt[t_1(x,w)|w \in A] \times \Pr[w \in A] + \expt[t_2(x,w)|w \in \overline{A}] \times \Pr[w \in \overline{A}] \numberthis \label{eq:min-bound} \\
			& \le \frac{\sum_{w \in A}t_1(x,w)}{|A|}\times \frac{|A|}{2^{|w|}} + \frac{\sum_{w \in \overline{A}}t_2(x,w)}{|\overline{A}|}\times \frac{|\overline{A}|}{2^{|w|}}\\
			& \leq \frac{2^{|w|} \times n^c}{|A|} \times \frac{|A|}{2^{|w|}} + \frac{2^{|w|} \times n^c}{|\overline{A}|}\times \frac{|\overline{A}|}{2^{|w|}}  = O(n^c)\numberthis \label{eq:bound}
		\end{align*}
		Note that Eq.~\ref{eq:min-bound} follows as $t(x,w)$ is the minimum among $t_1(x,w)$ and $t_2(x,w)$ and Eq.~\ref{eq:bound} follows from Lemma \ref{cat-configs} where $c$ is some absolute constant. Thus it follows that $\expt[t(x,w)] \le \poly(n)$.  Hence, the overall running time of $M'$ will be polynomial on expectation. 
 		
		We now argue that for any $A \in \L$, $\ACL(A) \cap \ACL(\overline{A}) = \CL$.
		
		For any $L \in \CL$, $L \in \ACL(A) \cap \ACL(\overline{A})$ as any catalytic machine always restores the catalytic content (irrespective of the choice of $A$). On the other hand, suppose that $L \in \ACL(A) \cap \ACL(\overline{A})$ via an almost-catalytic machine $M_1$ with restoration for $A$ and via an almost logspace catalytic machine $M_2$ with restoration for $\overline{A}$. Since $A$ can be decided in logspace, the catalytic algorithm first checks if the catalytic content belongs to $A$ and runs $M_1$ and runs $M_2$ otherwise. The resulting machine is indeed catalytic as it restores irrespective of the catalytic content and uses only logarithmic work space. Hence $L \in \CL$.
	\qed \end{proof}

%


\section{Almost-catalytic Computation via Error Correcting Codes}
    \label{sec:almost-catalytic-basic}
	We observed for any $A \subseteq \Sigma^*$, $\ACL(A) \subseteq \PSPACE$. We now show that there exists $A \subseteq \Sigma^*$ such that $\PSPACE \subseteq \ACL(A)$. We prove this by showing that there exists an $A \subseteq \Sigma^*$ such that for any $k$ and for any $L \in \DSPACE(n^k)$, $L \in \ACL(A)$. This suffices since $\PSPACE = \bigcup_{k \ge 0} \DSPACE(n^k) \subseteq \ACL(A)$.
	
	Our intuition is the following : 
	any computation can be seen as ``corrupting'' the catalytic tape content making the restoration difficult. With this view, it is natural to set $A$ to be codewords from an error correcting code of good distance. In addition, the code should be decodable in $O(\log n)$ space. 
	In the following Theorem, we choose $A$ to be one such code. 
	
    \almostCatalytic*    
	
	\begin{proof}
		Fix any $k\ge 0$. Let $L \in \DSPACE(n^k)$ via a Turing machine $M$ using a work space of $c n^k$ for some constant $c >0$. The goal is to construct a catalytic logspace Turing machine $M'$ such that $L(M')=L$ and it always restores the catalytic content $w$ if $w \in A$. We choose our $A$ such that $A_n := A \cap \{0,1\}^n$ consists of codewords of an explicit $[n,\frac{n}{4},\alpha n]_2$ linear code constructed by Spielman~\cite{linearSpielman}. Here, $\alpha$ (a constant, independent of $n$) is the relative distance of the code (as described in Theorem 19 of~\cite{linearSpielman}). In their work, it was shown that these codes can be decoded in deterministic logspace. Let $D$ be such a logspace decoding machine.
  
		We now describe a machine $M'$ (shown in Algorithm \ref{code_alg}) accepting $L$. With $x \in \Sigma^*$ of length $n$, let the length of the catalytic tape be $bn^k$ where $b$ is at least $\frac{2c}{\alpha}$. We work with the set $A_{bn^k}$ (where $A_n$ is as defined above). Let $D$ be the logspace decoder, given access to a string of length $bn^k$, can correct it using $O(\log n)$ space provided the string is within the decoding limit of some codeword in $A_{bn^k}$.

        \begin{algorithm}
			\caption{\textbf{Description of $M'$} on input $x$ and initial catalytic tape content $w$}\label{code_alg}
			\begin{algorithmic}[1]
				\STATE On input $x$, run $M$ on $x$ using the first $cn^k$ cells of the catalytic tape as the work tape for $M$.
				\STATE  Using the work tape as the work space for $D$, decode the content of the catalytic tape.
                \STATE Accept if $M$ accepted $x$
                \STATE Else reject
			\end{algorithmic}
		\end{algorithm}

			
			
		
		Since $M'$ simulates $M$, $L(M')=L(M) = L$. Let $w$ be the initial content of the catalytic tape with $|w|=bn^k$. Let $w'$ be its content at the end of the computation in Step 1 of $M'$. Observe that $w$ and $w'$ can differ in at most $cn^k$ bits as $M$ uses only the first $cn^k$ bits of the catalytic tape. If $w \in A_{bn^k}$, then $w$ is a codeword and $w'$ must fall in the Hamming ball of radius $\frac{\alpha }{2}bn^k$ since $\Delta(w',w) \le cn^k \le \frac{\alpha }{2}bn^k$ by the choice of $b$. Hence upon running $D$ on $w'$ in Step 2 of $M'$, will restore the catalytic tape content to $w$. Observe that this step uses $O(\log n)$ work tape cells. Thus, the above arguments imply that $L \in \ACL(A)$ via the machine $M'$. 



        The lower bound on $\calR(A_k)$ follows \cref{lem:codes} since the set $A_k$ is exactly the set of codewords of Spielman codes that have $\delta = O(1)$ \cite{linearSpielman}. 
        The lower bound on $\calP(A_k)$ follows from the minimum distance of the set $A_k$ is $d > 1$, and hence no two elements of $A_k$ can be covered by the same subcube. Hence $\calP(A_k) \ge |A_k|$ which is at least $2^{m/4}$. 
	\qed \end{proof}
\vspace{-3mm}
    We remark that it also suffices if the length of the codewords is a polynomial in $n$ (message length) and has a good distance that is logspace constructible and decodable. In addition to the Speilman codes~\cite{linearSpielman}, logspace decodable codes from \cite{GK06} also suffice for the above theorem. 

    
\section{An Improvement on the Subcube Partition Complexity of the Catalytic Set}\label{sec:limalmostcatalytic}
In \cref{Spthm}, we showed that any $\PSPACE$ algorithm can be simulated in almost-catalytic logspace by restoring catalytic content $w$ that are codewords of a carefully defined code as the set $A$. The ideal case would be to cover every such $w$ that appears as catalytic content. With this motivation, in the main result of this section (\cref{thm:code-improved-catalytic}), we attempt to cover strings that are not codewords as well at the expense of using less space. This allows the set $A$ to be larger than the one in~\cref{Spthm} and also has a better subcube partition complexity. 


\codeCatalytic*
\begin{proof}
        Let $L \in \DSPACE(\log^k n)$ via a Turing machine $M$. Let $m=n^k$ and $C$ be an $[m,m/4,\alpha m]_2$ logspace decodable Spielman code where $\alpha>0$ is a constant~\cite{linearSpielman}. We crucially use the existence of the family of functions ($f_m$) (Theorem 19 of \cite{linearSpielman}) in our algorithm defined as:
	  $f_m\colon \{0,1\}^{m} \times \{0,1\}^{\log m} \rightarrow \{0,1\}$ is a Boolean      function that takes in word $w \in \{0,1\}^m$ such that $\exists y \in \{0,1\}^{m/4}$ with $d(C(y),w) \le \frac{\alpha m - 1}{2}$ and an index $j \in [m]$ and outputs 1 if and only if $w_j \ne C(y)_j$. If $w_j \ne C(y)_j$, then we shall denote them as corrupted bits/indices of $w$. Here, $d(\cdot,\cdot)$ denotes the Hamming distance between two binary strings.

      In addition, the function $f_m$ can be computed by a log-space uniform family of bounded fan-in log-depth polynomial size circuits. Note that these circuits can be evaluated in $O(\log m) = O(\log n)$ space. Hence for any given $w \in \{0,1\}^m$ and $j \in [m]$, $f(w,j)$ can be computed in $O(\log n)$ space.
      
    We define $T_m$ to be a subset of all strings that are uniquely decodable to some codeword in $C$. 
    More precisely, $T_m = \left \{z \in \{0,1\}^m \mid \exists y \in \{0,1\}^{m/4}, d(z,C(y)) \le \Delta \right \}$ where $\Delta:=\frac{m}{\log^kn} \times \frac{\log n}{\log(\log^k n)}$. Define $A_k = \bigcup_{m\ge 1} T_m$. 
      The description of an $\ACL$ machine $M'$ simulating $M$ is given in Algorithm~\ref{alg:logk}.
		
		\begin{algorithm}
			\begin{algorithmic}[1] 
				\STATE Partition $[m]$ into disjoint contiguous blocks $B_1, B_2, \ldots B_\ell$ each of size $b = \log^kn$.
				\STATE Using the function $f_m$, find the first $i \in [\ell]$ in $w$ such that 
                    $|\{ j \in B_i \mid f_m(w,j) = 1 \} | \le \frac{\log n}{\log(\log^kn)}$
                    \STATE If such an $i$ does not exist, set $i=\ell$, $E=\emptyset$, and go to line 6. \hfill //in this case $w \notin A_k$ 
				\STATE Store the start and end indices of the block $B_i$. Call them $p$ and $q$ respectively.
                    \STATE Let $E = \{ i \in [b] \mid f_m(w,p+i) = 1 \}$ be the corrupted bits of $B_i$.
				\STATE Run $M$ on $x$ using catalytic tape cells indexed by $B_i$ as work space for $M$ and accept if and only if $M$ accepts. 
                    \STATE {\bf Restoration:} Let $w'$ be the content of the catalytic tape at the end of the computation. \\ For each $j \in B_i$, if $[f_m(w',j) = 1)] \oplus [j \in E]$, flip $w'_j$.
				\caption{Description of $M'$ on input $x$ and catalytic tape content $w$ with $|w|=m$.}
                \label{alg:logk}
			\end{algorithmic}
		\end{algorithm}
				

    Firstly, we argue the correctness of the above algorithm. To see this, irrespective of whether $w \in A_k$ or not, steps 2 and 3 of \cref{alg:logk} will find a block $B_i$ of size $O(\log^k n)$ such that the cells of the catalytic tape indexed by $B_i$ will be used to correctly simulate $M$ on $x$ which requires only $O(\log^k n)$ space. 
    
    We argue now that the above algorithm restores $w$ at the end of the computation if $w \in A_k$. If $w \in A_k$,  then there is a codeword $\gamma \in \{0,1\}^m$ within a Hamming distance of $\Delta=\frac{m}{\log^kn} \times \frac{\log n}{\log(\log^k n)}$  from $w$. Since there are $\ell = \frac{m}{\log^k n}$ blocks, by averaging, there must be a block with at most $\frac{\log n}{\log(\log^k n)}$ errors. Let $B_i$ be the first such block. Note that step 2 of \cref{alg:logk} will indeed find such a $B_i$.
      		
        
       Since $|B_i| \le O(\log^k n)$, we still have that $d(w',\gamma) \le d(w,\gamma) + O(\log^k n) \le \frac{m}{\log^kn} \times \frac{\log n}{\log(\log^k n)} + O(\log^k n) \le \frac{\alpha m -1}{2}$ for large enough $n$. Hence the word $w'$ is still within the decoding radius of the codeword that we started with.
		
	   Recall that $w$ and $w'$ are content of the catalytic tape, respectively, at the beginning and end of step $6$. For $j \notin B_i$, step 6 does not change $w_j$ and hence $w'_j = w_j$. We show that the bits indexed by $B_i$ get restored. If $j \in B_i$, then $w_j$ may get changed during the simulation of the machine $M$ in step 6. Recall that the step $5$ computes the set of corrupted indices in block $B_i$ as the set $E$. For a $j \in B_i$, there are two cases:
       \begin{description}
           \item{\bf Case 1:} $w'_j = \gamma_j$. This implies that $j$-th bit in $w_j'$ is not corrupted. If in addition, $j \in E$, we have $w_j \ne \gamma_j$ which implies $w_j \ne w'_j$. Hence when the algorithm flips $w_j'$ in line 7, it makes it equal to $w_j$. For $j \notin E$, we have that $w_j = \gamma_j = w'_j$, and hence no flipping is required in line 7 of the algorithm.
           \item{\bf Case 2:} $w'_j \ne \gamma_j$. This implies that $j$-th bit in $w_j'$ is corrupted. If in addition, $j \not\in E$ we have $w_j = \gamma_j$ which implies $w_j \ne w'_j$. Hence when the algorithm flips $w_j'$ in line 7, it makes it equal to $w_j$. In case, $j \in E$, we have that $w_j \ne \gamma_j$ and hence $w_j = w'_j$. Hence no flipping is required in line 7 of the algorithm.
       \end{description}
 
       We now argue the space bound for $M'$. The indices $p, q \in [m]$ as well as $i, j \in [\ell]$ can be stored in $O(\log n)$ space. Note that we store $E \subseteq [b]$, taking $|E| \log b$ many bits. Since the number of corrupted bits $|E|$, is at most $\frac{\log n}{\log(\log^k n)}$, we can store $E$ in the work tape using $O(\log n)$ bits. As mentioned above, the function $f_m$ can also be computed in $O(\log n)$ space as needed in lines 2, 5 and 7 of the algorithm. This establishes the space bound.

        The lower bound on $\calR(A_k)$ follows \cref{lem:codes} since the set $A_k$ is exactly the set of codewords of Spielman codes that have $\delta = O(1)$ \cite{linearSpielman} and that $R_\epsilon(A)$ is a monotone property with respect to $A$. We now prove the improvement on $\calP(A_k)$.

The lower bound on $\calP(A_k)$ follows from a careful analysis of the Fourier coefficients of the Threshold function (See ~\cref{lem:partition-complexity-ball-main} presented in \cref{app:sec:lbunion}) noting that the set $A_k$ is defined to the union of Hamming balls of radius $\Delta=\frac{m}{(\log^{k-1}n)\log(\log^k n)}$ and that for large enough $m$, $\Delta < m/2-\sqrt{m}$.
	\qed \end{proof}

    \begin{appsection}{Omitted Proofs from \cref{ssec:lb-hamming-balls} : A Lower Bound on the partition complexity for Union of Hamming Balls}{app:sec:lbunion}
    This section details the proofs of statements given in \cref{ssec:lb-hamming-balls}. For convenience of the reader, some of the terminologies defined there are repeated in this section.
    
    For two strings, $x,y \in \{0,1\}^m$, the Hamming distance, is denoted by $\Delta(x,y)$. The same definition can be extended to subsets as follows: for any $A,B \subseteq \{0,1\}^m$, $\Delta(A,B) = \min\{\Delta(a,b) \mid a \in A, b \in B \}$.
    
    A set $H \subseteq \{0,1\}^m$ is said to be a Hamming ball if and only if there exists a $k \ge 0$ and a $z \in \{0,1\}^n$ such that for every $h \in H$, $\Delta(h,z) \le k$. We call $k$ as the \textit{radius} of the Hamming ball $H$ and $z$ to be its \textit{center}. 
    
    \begin{proposition}\label{prop:partition-complexity-balls-union}
        Let $A \subseteq \{0,1\}^*$ be such that for any $m \ge 1$, $A_m := A \cap \{0,1\}^m$ can be expressed as a union of Hamming balls $H_1, H_2, \ldots, H_t$ over $\{0,1\}^m$ such that for any $i\ne j$, $\Delta(H_i, H_j) > 1$. Then, $\calP(A_m) = \sum_{i=1}^t \calP(H_i)$.
    \end{proposition} 
    \begin{proof}
        Consider any partition of $A_m$ into $t$ subcubes given by $C_1, C_2, \ldots, C_t$. Suppose that there exists a subcube $C_k$, such that it contains points from $H_i$ and $H_j$ for some $i \ne j$. As it is a partition, every point in the subcube must belong to some Hamming ball and hence there exists two strings $x,y \in C_k$ that differ in one bit with $x \in H_i$ and $y \in H_j$. But this contradicts $\Delta(H_i,H_j) > 1$. Hence, each $C_i$ can contain at most one Hamming ball. With no sub cube partition of $A_m$ intersecting two Hamming balls, we conclude $\calP(A_m) \ge \sum_{i=1}^t \calP(H_i)$.
    
        On the other hand, since a sub cube partition of the Hamming balls gives a sub cube partition of $A_m$, $\calP(A_m) \le \sum_{i=1}^t \calP(H_i)$.
    \qed \end{proof}
    We now set up the necessary tools to obtain a lower bound on the subcube partition complexity of Hamming balls. \cite{CKLS16} showed that the partition complexity of a set viewed as a Boolean function is lower bounded by the sum of absolute values of its Fourier coefficients.

    Define the Boolean function $Th_{m,k} \colon \{0,1\}^m   \to \{-1,1\}$ as for any $x \in \{0,1\}^m, Th_{m,k}(x) = -1$ if $|x| \le k$ and $1$ if $|x| > k$.  
	 We start with the observation that a Hamming ball of radius $k$ centered at $0^m$  are precisely the set of inputs on which the threshold Boolean function $Th_{m,k}$ evaluates to $-1$. 
    
    In \cref{prop:fc-threshold} and \cref{prop:fc-values-more-than-k}, we obtain closed-form expressions for Fourier coefficients of $Th_{n,k}$. 

	\begin{proposition}
		For any $1 \le i \le k < n/2$ and $S \subseteq [n]$ with $|S|=i$,  \label{prop:fc-threshold} $\widehat{Th_{n,k}}(S) = 
                         \frac{1}{2^{n-1}} \binom{n-2i+1}{k-i}$  
        
	\end{proposition}
	\begin{proof}
		Since the Threshold function is symmetric, $\widehat{Th_{n,k}}(S)$ is the same for all $S \subseteq [n]$ of size $i$. Hence, without loss of generality, we assume that $S = \{1,2,\ldots,i\}$. 

		We argue this by induction on $n+|S|$. For the base case $n+1$, let $S \subseteq [n]$ with $S=\{1\}$. Then, $\widehat{Th_{n,k}}(\{ 1\})$ is the fraction of edges in the Boolean hypercube such that it evaluates to different values at the two ends.  Hence, $\widehat{Th_{n,k}}(\{ 1\}) = \frac{n\binom{n-1}{k-1}}{n2^{n-1}} = \frac{\binom{n-1}{k-1}}{2^{n-1}}$ as desired.
		
		For the induction case, let $n+|S| = n+j$ with $S = \{1,2,\ldots,j\}$ for some $1\le j \le k$. Suppose that the result holds for all values strictly  smaller than $n+j$. We now claim that $\widehat{Th_{n,k}}(S) = \frac{1}{2}( \widehat{Th_{n-1,k}}(S\setminus\{j\}) - \widehat{Th_{n-1,k-1}}(S\setminus\{j\}))$. To see this, 
		\begin{align*}
			\widehat{Th_{n,k}}(S) & = \frac{1}{2^n} \sum_{x \in \{0,1\}^n} Th_{n,k}(x) (-1)^{x_1+\ldots +x_j} \\
			& = \frac{1}{2^n}\left ( \sum_{x \in \{0,1\}^n~:~ x_j=0} Th_{n,k}(x) (-1)^{x_1+\ldots +x_{j-1}} -  \sum_{x \in \{0,1\}^n~:~ x_j=1} Th_{n,k}(x) (-1)^{x_1+\ldots +x_{j-1}} \right )\\
			& = \frac{1}{2^n}\left ( \sum_{x \in \{0,1\}^{n-1} } Th_{n-1,k}(x) (-1)^{x_1+\ldots +x_{j-1}} -  \sum_{x \in \{0,1\}^{n-1} } Th_{n-1,k-1}(x) (-1)^{x_1+\ldots +x_{j-1}} \right )\\
			& = \frac{1}{2} \left (\widehat{Th_{n-1,k}}(S\setminus \{j\}) - \widehat{Th_{n-1,k-1}}(S\setminus \{j\})  \right ).
		\end{align*}
		
		Applying induction hypothesis, we get $\widehat{Th_{n,k}}(S)$ as
		\[ \frac{1}{2 \cdot  2^{n-1-1}} \left (\binom{n-1-2(j-1)+1}{k-(j-1)} - \binom{n-1-2(j-1)+1}{k-1-(j-1)} \right ) = \frac{1}{2^{n-1}} \binom{n-2j+1}{k-j} \]
		The last equality follows, since for positive integers $m$ and $r$, $\binom{m}{r} - \binom{m}{r-1} = \binom{m-1}{r-1}$. This completes the induction.
	\qed \end{proof}
    We now consider the case when $i > k$.
    \begin{proposition} \label{prop:fc-values-more-than-k}
		For any $1 \le i < n/2$ and $k < n/2$, $S \subseteq [n]$ with $|S|=k+i$, the value of $\widehat{Th_{n,k}(S)}$ is as follows. \label{prop:fc-all-threshold} 
        \begin{enumerate}[(a)]
            \item For $i=1$, with $|S| = k+1$
		 \[ \widehat{Th_{n,k}}(S) = \begin{cases} 
                         0  & \text{ if } k \text{ is even} \\
                         1/2^{n-1} & \text{ if } k \text{ is odd}
                         \end{cases}
         \] 
         \item For $i = 2$ with $|S| = k+2$
        \[ \widehat{Th_{n,k}}(S) = \begin{cases} 
                         -\frac{k/2}{2^{n-1}}  & \text{ if } k \text{ is even} \\
                         \frac{(k+1)/2}{2^{n-1}} & \text{ if } k \text{ is odd}
                         \end{cases}
         \]
        
        \item For $i \ge 2$ with $|S|=k+i$, \[ |\widehat{Th_{n,k}}(S)| \ge \frac{k/2}{2^{n-1}}     \]
        \end{enumerate}
	\end{proposition}
    \begin{proof}
        \textbf{Part (a)}:  By induction on $k$. There are two base cases to consider: $k=1$ and $k=2$. 
        
        For the base case, $k=1$ we have $|S|=2$. Without loss of generality, let $S= \{1,2\}$. Then, $\widehat{Th_{n,1}}(\{1,2\}) = \frac{1}{2} \left ( \widehat{Th_{n-1,1}}(\{1\}) - \widehat{Th_{n-1,0}}(\{1\}) \right )$. The first term is $1/2^{n-2}$ by \cref{prop:fc-threshold} and the second term is zero as $Th_{n-1,0}$ is a constant function. Hence, $\widehat{Th_{n,1}}(\{1,2\}) = 1/2^{n-1}$. 
        
        For $k=2$ we have $|S|=3$. Without loss of generality, let $S= \{1,2,3\}$. Then, $\widehat{Th_{n,2}}(\{1,2,3\}) = \frac{1}{2} \left ( \widehat{Th_{n-1,2}}(\{1,2\}) - \widehat{Th_{n-1,1}}(\{1,2\}) \right )$. The first term is $1/2^{n-2}$ by \cref{prop:fc-threshold} and the second term is also $1/2^{n-2}$ by case $k=1$. Hence, $\widehat{Th_{n,2}}(\{1,2,3\}) = 0$ as desired. 

        We are now ready to argue the induction case. Suppose $k$ is odd with $S = \{1,2,\ldots,k+1\}$. Then,
        \begin{equation} \label{eq:fc-split}
         \widehat{Th_{n,k}}(S) = \frac{1}{2} \left ( \widehat{Th_{n-1,k}}(S\setminus \{k+1\}) - \widehat{Th_{n-1,k-1}}(S \setminus \{k+1\}) \right ).   
        \end{equation}
          By \cref{prop:fc-threshold}, the first term of~\cref{eq:fc-split} is $1/2^{n-2}$ and the second term, by induction is $0$ (as $k-1$ is even). Hence, $\widehat{Th_{n,k}}(S) = 1/2^{n-1}$.

        For the case of an even $k$, the first term of~\cref{eq:fc-split} is $1/2^{n-2}$ as before and the second term is also $1/2^{n-2}$ by induction (as $k-1$ is odd). Hence, $\widehat{Th_{n,k}}(S) = 0$. This completes the proof of part (a).
        
        \textbf{Part (b)}: By induction on $k$. There are two base cases $k=1$ and $k=2$ similar to Part (a).
        
        For the base case, $k=1$, we have $|S|=3$. Without loss of generality, let $S=\{1,2,3\}$. Then, $\widehat{Th_{n,1}}(\{1,2,3\}) = \frac{1}{2} \left( \widehat{Th_{n-1,1}}(\{1,2\}) - \widehat{Th_{n-1,0}}(\{1,2\}) \right )$. The first term is $1/2^{n-2}$ by part (a) and the second term is $0$. Hence, $\widehat{Th_{n,1}}(\{1,2,3\}) = \frac{1}{2^{n-1}}$ as desired.

        For the base case, $k=2$, we have $|S|=4$. Without loss of generality, let $S=\{1,2,3,4\}$. Then, $\widehat{Th_{n,2}}(\{1,2,3,4\}) = \frac{1}{2} \left( \widehat{Th_{n-1,2}}(\{1,2,3\}) - \widehat{Th_{n-1,1}}(\{1,2,3\}) \right )$. The first term is $0$ by part (a). The remaining term is $-\frac{1}{2}\widehat{Th_{n-2,1}}(\{1,2\})$ (by a similar reasoning as done for Part (b) $k=1$). By Part (a) base case $k=1$, $\widehat{Th_{n-2,1}}(\{1,2\})$ is $1/2^{n-3}$. Hence, $\widehat{Th_{n,2}}(\{1,2,3,4\})=-\frac{1}{2^{n-1}}$ as desired.
        
        We are now ready to argue the induction case. Suppose $k$ is odd with $S = \{1,2,\ldots,k+2\}$. Then,
        \begin{equation} \label{eq:fc-split-partb}
         \widehat{Th_{n,k}}(S) = \frac{1}{2} \left ( \widehat{Th_{n-1,k}}(S\setminus \{k+2\}) - \widehat{Th_{n-1,k-1}}(S \setminus \{k+2\}) \right ).   
        \end{equation}

        By Part (a), the first term of~\cref{eq:fc-split-partb} is $1/2^{n-2}$ (as $k$ is odd) and the second term, by induction hypothesis, is $-\frac{(k-1)/2}{2^{n-2}}$ (as $k-1$ is even). Hence, $\widehat{Th_{n,k}}(S) = \frac{(k+1)/2}{2^{n-1}}$ as desired.

        For the case of an even $k$, the first term in~\cref{eq:fc-split-partb} is $0$ (as $k$ is even) and the second term, by induction hypothesis, is $\frac{k/2}{2^{n-2}}$. Hence, $\widehat{Th_{n,k}}(S) = -\frac{k/2}{2^{n-1}}$ as desired. This completes the proof of part (b).

        \textbf{Part (c)}: We argue the following for $|S|=k+i$ with $i \ge 2$.
        \[ \widehat{Th_{n,k}}(S) \begin{cases}  
        \le - \frac{k/2}{2^{n-1}}  & \text{ if $k$ is even} \\
        \ge \frac{k/2}{2^{n-1}}  & \text{ if $k$ is odd} 
        \end{cases}
        \]
        
        Proof is by induction on $|S|$. The base case of $i=2$ holds by Part (b). Suppose $S = \{1,2,\ldots,k+i\}$ for some $i \ge 2$. Then,

        \begin{equation} \label{eq:fc-split-partc}
         \widehat{Th_{n,k}}(S) = \frac{1}{2} \left ( \widehat{Th_{n-1,k}}(S\setminus \{k+i\}) - \widehat{Th_{n-1,k-1}}(S \setminus \{k+i\}) \right ).   
        \end{equation}

        Suppose $k$ is odd. Consider \cref{eq:fc-split-partc}. For the first summand, with $S'= S\setminus \{k+i\}$ of size strictly smaller than $k+i$ has $k'=k$  which is odd. By induction, $\widehat{Th_{n-1,k'}}(S') = \widehat{Th_{n-1,k}}(S\setminus \{k+i\}) \ge \frac{k/2}{2^{n-2}}$. For the second summand, the set $S' = S\setminus \{k+i\}$ has $k'=k-1$ which is even. With $k'$ being even, induction tells that $\widehat{Th_{n-1,k'}}(S')=\widehat{Th_{n-1,k-1}}(S\setminus \{k+i\}) \le -\frac{(k-1)/2}{2^{n-2}}$. Hence, by \cref{eq:fc-split-partc}, 
         \[ \widehat{Th_{n,k}}(S) \ge \frac{1}{2} \left ( \frac{k/2}{2^{n-2}} +\frac{(k-1)/2}{2^{n-2}} \right )
         = \frac{k-1/2}{2^{n-1}} \ge  \frac{k/2}{2^{n-1}}\]
        The last inequality follows since $k \ge 1$. 

        Suppose $k$ is even. For the first summand, with $S'= S\setminus \{k+i\}$ of size strictly smaller than $k+i$ has $k'=k$ which is even. By induction,  $\widehat{Th_{n-1,k'}}(S')=\widehat{Th_{n-1,k}}(S\setminus \{k+i\}) \le -\frac{k/2}{2^{n-2}}$. For the second summand, the set $S' = S\setminus \{k+i\}$ with $k'=k-1$ which is odd. With $k'$ being old, induction tells that $\widehat{Th_{n-1,k'}}(S')=\widehat{Th_{n-1,k-1}}(S\setminus \{k+i\}) \ge \frac{(k-1)/2}{2^{n-2}}$. Hence, by \cref{eq:fc-split-partc}, 
         \[ \widehat{Th_{n,k}}(S) \le \frac{1}{2} \left ( -\frac{k/2}{2^{n-2}} -\frac{(k-1)/2}{2^{n-2}} \right )
         = \frac{-k+1/2}{2^{n-1}} \le  -\frac{k/2}{2^{n-1}}\]
        The last inequality follows since $k \ge 1$. This completes the induction.
    \qed \end{proof}
    
    We now argue below (in \cref{prop:partition-complexity-ball}) that any Hamming ball over $\{0,1\}^m$ with radius $k$ strictly less than $m/2-\sqrt{m}$ centered at $0^m$, must have a subcube partition complexity of $\Omega(k)$.

    \begin{proposition}\label{prop:partition-complexity-ball}
        Let $H$ be a Hamming ball over $\{0,1\}^m$ of radius $k < m/2-\sqrt{m}$ centered at $0^m$. Then $\calP(H) = \Omega(k)$.
    \end{proposition}
    \begin{proof}

	  It is known that the partition complexity of a Boolean function $f$ on $m$ bits is lower bounded by $\sum_{S \subseteq [m]} |\widehat{f}(S)|$ where $\widehat{f}(S)$ is the Fourier coefficient of $f$ (cf. Lemma 3.8 of \cite{CKLS16}). 
    Hence to lower bound the partition complexity of a Hamming ball, it suffices to compute $\sum_{S} |\widehat{Th_{m,k}}(S)|$.
    
     Every $S \subseteq [m]$ of the same size $i$, has the same value for $|\widehat{Th_{m,k}}(S)|$ which is at least $k/2$ by \cref{prop:fc-values-more-than-k} when $i \ge k+2$. Hence, 
     \begin{align*} 
     \sum_{S} |\widehat{Th_{m,k}}(S)|  & \ge \sum_{i=k+2}^{m/2-1} \sum_{S ~:~ |S| = i} |\widehat{Th_{m,k}}(S)| \ge \sum_{i= k+2}^{m/2-1} \frac{1}{2^{m-1}} \binom{m}{i} \frac{k}{2}  \\
     & \ge \sum_{i= m/2-\sqrt{m}}^{m/2-1} \frac{1}{2^{m-1}} \binom{m}{i} \frac{k}{2} = \frac{k}{2^{m}} \sum_{i= m/2-\sqrt{m}}^{m/2-1} \binom{m}{i} = \Omega(k)
     \end{align*}
     The binomial sum in the last inequality can be shown to be a constant fraction of $2^m$ by a standard application of Chebyshev's inequality. This yields the desired lower bound of $\Omega(k)$.
    \qed \end{proof}
    \begin{lemma}
    \label{lem:partition-complexity-ball-main}
        Let $A \subseteq \Sigma^*$ such that for every $m \ge 1$, $A_m$ is a disjoint union of Hamming balls $H_1, \ldots, H_t$ of radius $k<m/2-\sqrt{m}$ over $\{0,1\}^m$ such that for every $i,j \in [t]$, $\Delta(H_i, H_j) > 1$. Then for every $m\ge 1$, $\calP(A_m) = \Omega(tk)$.
    \end{lemma}

    \begin{proof}
    For a contradiction, suppose that $\calP(A_m) = o(tk)$. By \cref{prop:partition-complexity-balls-union}, which says $\calP(A_m) = \sum_{i=1}^t \calP(H_i)$, there exists an $H_i$ centered at some $z \in \{0,1\}^m$ such that $\calP(H_i) = o(k)$. Consider the set $H':=H_i\oplus z = \{h \oplus z \mid h \in H_i\}$ obtained by taking the bitwise XOR of each string in $H_i$ by $z$.  Observe that $H'$ is a Hamming ball centered at $0^m$ of same radius as that of $H_i$ since the operation performed does not alter the relative Hamming distance between the points. With the subcubes shifted by $z$ also forming a partition of $H'$, we have $\calP(H') = o(k)$ which contradicts \cref{prop:partition-complexity-ball}. This completes the proof.
    \qed \end{proof}
        
    \end{appsection}

\paragraph*{Two Limitations of the Approach towards $\DSPACE(\log^k n) \subseteq \ACL(\Sigma^*)$}
    We first note a limitation of the approach due to the fact that there is a direct simulation of the $\DSPACE(\log^k n)$ machine in the argument. \cite{BCKLS14} observed that there cannot be a step-by-step simulation of Turing machines that use $\omega(s(n))$ space by using catalytic Turing machines that uses $s(n)$  work space and even $2^{s(n)}$ catalytic space. We note that this also implies a limitation of our approach towards almost-catalytic Turing machines as well.
 
    Consider the following family of algorithms that attempts to show $\DSPACE(\log^k(n)) \subseteq \CL$ via error correcting codes as follows: Let $L$ belongs to $\DSPACE(\log^k(n))$ via machine $M$. Then we construct a catalytic  machine as follows:
        (1)~Apply logspace computable transformations to the initial catalytic tape content to make it ``recoverable'' from $O(\log^k n)$ errors. 
        (2)~Run the machine $M$ on input $x$ on the catalytic tape.
        (3)~Correct the $O(\log^k n)$ errors on the catalytic tape and restore $w$.
        (4)~Accept if $M$ accepts and reject otherwise.
            
	
    
    \begin{proposition}\label{prop:simulation}
        There is no simulation of deterministic polylogarithmic space in catalytic logspace via direct simulation and using logspace decodable error correcting codes.
    \end{proposition}
    \begin{proof}
	We argue that the direct simulation cannot work as it is. Indeed, it implies that the machine $M$ must necessarily run in expected polynomial time (with respect to choice of initial catalytic content). In other words, every $O(\log^k n)$ machine must run in polynomial time - a statement which can be proved to be false. We argue the same below.
 
	Consider the case when machine $M$ is constructed as follows: $M$ visits all its configurations before halting. There are $O(\exp(\log^{k} n))$ many configurations to the machine and thus it takes time at least $O(\exp(\log^{k} n))$ to run. In step $2$ of our algorithm, we run the machine $M$ directly (i.e. step-by-step). So our algorithm too runs in time at least $O(\exp(\log^{k} n))$. The initial content $w$ does not affect step $2$, so the average running time (over the choice of $w$) is still super-polynomial which is a contradiction to the fact that any $\CL$ machine takes polynomial running time on average over the choice of $w$.
	\qed \end{proof}


 We observe a second limitation of the approach due to the fact that we cannot expect linear codes to have a covering radius as low as required for the algorithm. More precisely, for the approach, we need the covering radius of the code $C \subseteq \{0,1\}^m$ to be at most $\frac{m}{(\log m)^{k-1} \log \log n}$. However, for every code with rate $r$, the covering radius is known\cite{CKMS85} to be at least $m\left(\frac{1}{2}-\frac{\sqrt{r}}{2^{3/2}}\right)$ which is at least $\Omega(n)$ even for constant rate codes. 



\noindent
\paragraph*{Acknowledgments} We would like to thank the anonymous reviewers for pointing out that \cref{thm:ACL-ZPP-L} works for any $A \subseteq \Sigma^*$ (previous versions stated restrictions on $A$), for pointing out Proposition~\ref{thm:code-improved-catalytic-reviewer} and for pointing an issue in earlier proof of \cref{prop:partition-complexity-ball}.

\bibliographystyle{splncs04}
\bibliography{refs}

\begin{thebibliography}{10}
\providecommand{\url}[1]{\texttt{#1}}
\providecommand{\urlprefix}{URL }
\providecommand{\doi}[1]{https://doi.org/#1}

\bibitem{arora-barak}
Arora, S., Barak, B.: Computational Complexity: A Modern Approach. Cambridge University Press, USA, 1st edn. (2009)

\bibitem{BCKLS14}
Buhrman, H., Cleve, R., Kouck\'{y}, M., Loff, B., Speelman, F.: Computing with a full memory: Catalytic space. In: Proceedings of the Forty-Sixth Annual ACM Symposium on Theory of Computing. p. 857–866. STOC '14, Association for Computing Machinery, New York, NY, USA (2014). \doi{10.1145/2591796.2591874}, \url{https://doi.org/10.1145/2591796.2591874}

\bibitem{BKLS18}
Buhrman, H., Kouck\'{y}, M., Loff, B., Speelman, F.: Catalytic space: Non-determinism and hierarchy. Theory of Computing Systems  \textbf{62}(1),  116–135 (jan 2018)

\bibitem{CKLS16}
Chakraborty, S., Kulkarni, R., Lokam, S.V., Saurabh, N.: Upper bounds on fourier entropy. Theoretical Computer Science  \textbf{654},  92--112 (2016). \doi{10.1016/j.tcs.2016.05.006}, computing and Combinatorics

\bibitem{CKMS85}
Cohen, G., Karpovsky, M., Mattson, H., Schatz, J.: Covering radius---survey and recent results. IEEE Transactions on Information Theory  \textbf{31}(3),  328--343 (1985). \doi{10.1109/TIT.1985.1057043}

\bibitem{CLMP24}
Cook, J., Li, J., Mertz, I., Pyne, E.: The structure of catalytic space: Capturing randomness and time via compression. Electron. Colloquium Comput. Complex.  \textbf{{TR24-106}} (2024), \url{https://eccc.weizmann.ac.il/report/2024/106}

\bibitem{CM24}
Cook, J., Mertz, I.: Tree evaluation is in space $o (\log n \cdot \log \log n)$. In: Proceedings of the 56th Annual ACM Symposium on Theory of Computing. p. 1268–1278. STOC 2024, Association for Computing Machinery, New York, NY, USA (2024). \doi{10.1145/3618260.3649664}, \url{https://doi.org/10.1145/3618260.3649664}

\bibitem{CMPDBS12}
Cook, S., McKenzie, P., Wehr, D., Braverman, M., Santhanam, R.: Pebbles and branching programs for tree evaluation. ACM Trans. Comput. Theory  \textbf{3}(2) (jan 2012). \doi{10.1145/2077336.2077337}, \url{https://doi.org/10.1145/2077336.2077337}

\bibitem{DGJST20}
Datta, S., Gupta, C., Jain, R., Sharma, V.R., Tewari, R.: Randomized and symmetric catalytic computation. In: Fernau, H. (ed.) Computer Science -- Theory and Applications. pp. 211--223. Springer International Publishing, Cham (2020)

\bibitem{Goldreich_2008}
Goldreich, O.: Computational Complexity: A Conceptual Perspective. Cambridge University Press (2008)

\bibitem{GJST19}
Gupta, C., Jain, R., Sharma, V.R., Tewari, R.: Unambiguous catalytic computation. In: Chattopadhyay, A., Gastin, P. (eds.) 39th {IARCS} Annual Conference on Foundations of Software Technology and Theoretical Computer Science, ({FSTTCS} 2019). LIPIcs, vol.~150, pp. 16:1--16:13 (2019)

\bibitem{GJST24}
Gupta, C., Jain, R., Sharma, V.R., Tewari, R.: Lossy catalytic computation (2024), \url{https://arxiv.org/abs/2408.14670}

\bibitem{GK06}
Guruswami, V., Kabanets, V.: Hardness amplification via space-efficient direct products. In: Proceedings of the 7th Latin American Conference on Theoretical Informatics. p. 556–568. LATIN'06, Springer-Verlag, Berlin, Heidelberg (2006)

\bibitem{KouckySurvey}
Kouck{\'y}, M.: Catalytic computation. Bull. EATCS  \textbf{118} (2016)

\bibitem{MertzSurvey}
Mertz, I.: Reusing space: Techniques and open problems. Bull. EATCS  \textbf{141} (2023)

\bibitem{AOBF}
O'Donnell, R.: Analysis of Boolean Functions. Cambridge University Press (June 2014), \url{http://dx.doi.org/10.1017/CBO9781139814782}

\bibitem{P23}
Pyne, E.: Derandomizing logspace with a small shared hard drive. Electron. Colloquium Comput. Complex.  \textbf{{TR23-168}} (2024), \url{https://eccc.weizmann.ac.il/report/2023/168}

\bibitem{linearSpielman}
Spielman, D.A.: The complexity of error-correcting codes. In: Chlebus, B.S., Czaja, L. (eds.) Fundamentals of Computation Theory. pp. 67--84. Springer Berlin Heidelberg, Berlin, Heidelberg (1997)

\end{thebibliography}

\ifthenelse{\equal{\movetoappendix}{1}}{
        \appendix
        \section{Appendix}
        \includecollection{appendix}
} { }

\newpage

\end{document}